\documentclass[conference]{IEEEtran}

\usepackage{cite}	
\usepackage{graphicx} 
\usepackage[latin1]{inputenc} 
\usepackage[T1]{fontenc} 
\usepackage{amsmath,amsfonts,amsbsy,amssymb} 
\usepackage{mathtools}
\usepackage{amssymb} 
\usepackage{amsthm}	
\usepackage{mathrsfs}
\usepackage[nolist]{acronym} 
\usepackage{tabularx} 
\usepackage{multirow}
\usepackage{wasysym}
\usepackage{float}
\usepackage[usenames,dvipsnames]{color}
\usepackage{soul}

\usepackage{enumitem} 

\newtheorem{proposition}{Proposition}
\newtheorem{lemma}{Lemma}
\newtheorem{theorem}{Theorem}

\newtheorem{remark}{Remark}

\def\rayleigh{\mathop{\mathrm{Rayleigh}}}
\def\exponential{\mathop{\mathrm{Exp}}}

\newenvironment{list4}{
  \begin{list}{$\bullet$}{
      \setlength{\itemsep}{0.05cm}
      \setlength{\labelsep}{0.2cm}
      \setlength{\labelwidth}{0.3cm}
      \setlength{\parsep}{0in}
      \setlength{\parskip}{0in}
      \setlength{\topsep}{0in}
      \setlength{\partopsep}{0in}
      \setlength{\leftmargin}{0.2in}}}
      {\end{list}}

\begin{document}
%
%
%
%
\title{\huge Power Allocation for Type-I ARQ Two-Hop Cooperative Networks for Ultra-Reliable Communication}
\author{\IEEEauthorblockN{Endrit Dosti, Themistoklis Charalambous and Risto Wichman}
	\IEEEauthorblockA{School of Electrical Engineering, Aalto University, Espoo, Finland \\ email: firstname.lastname@aalto.fi}}
\maketitle

%
%
%
%
\begin{abstract}		
We analyze the performance of amplify-and-forward (AF) automatic repeat request (ARQ) for a two-hop cooperative system with reliability constrains. For this setup, we first derive the closed-form outage probability expression. Next, we present a power allocation scheme that allows us to achieve a target outage probability, while minimizing the outage-weighted average power expenditure for asymmetric power allocation between the source and relay. This is cast as an optimization problem, and the optimal power allocation (OPA) is obtained in closed form by invoking the Karush-Kuhn-Tucker (KKT) conditions. We evaluate numerically the OPA strategy between different AF-ARQ transmission rounds and we show that the proposed scheme provides large power gains with respect to the optimized point-to-point ARQ scheme, as well as with respect to the equal power allocation (EPA) strategy.  
\end{abstract}
	

%
%
%
%
\section{Introduction}\label{sec:intro}

Ultra-Reliable Low-Latency Communication (URLLC) systems aim at guaranteeing successful data transmission within stringent delay requirements. Such characteristics are indispensable for various applications, such as, tele-surgery, intelligent transportation, and industry automation \cite{Popovski2017, Schulz}. Fading in wireless channels, which causes fluctuations in the received signal strength resulting in loss of transmitted packets \cite{Goldsmith}, is the main impediment towards achieving URLLC.

To alleviate the effects of this phenomenon, schemes that rely on \emph{packet retransmission} and \emph{cooperative communications} have been proposed in the literature. The problem of integrating retransmission and cooperative protocols, both jointly and separately, in traditional communications systems is a widely studied topic; see, for example, \cite{Lee, Makki, LTS, Tumula2, 5754756, Larsson_2, Larsson, Dosti2}. For instance, in \cite{Lee} the authors propose OPA schemes for maximizing throughput for cooperative decode-and-forward ARQ relaying schemes. While the cooperative scheme proposed in \cite{Lee} exhibits reasonably good performance, it requires that all relay nodes decode the information, which results in increased latency in the system. In \cite{Makki}, the authors analyze the performance of relay-ARQ networks under quasi-static and fast fading conditions. Therein, the authors show that OPA across different ARQ rounds provides limited gains with respect to the EPA strategy. However, their analysis is limited only in outage probabilities spanning the range $\epsilon \in [10^{-1}, 10^{-3}]$ and does not consider ultra-reliability constrains. In \cite{Tumula2}, the authors analyze and provide OPA strategy for Chase Combining (CC) Hybrid ARQ (HARQ). However, their analysis is limited to the case of maximum two transmissions. Similarly, in \cite{5754756}, the authors provide an OPA scheme for Incremental Redundancy (IR) HARQ limited only for maximum two transmissions. The analysis is then extended in \cite{Larsson_2}, where the authors propose an OPA strategy for the point-to-point IR HARQ scheme valid for any number of transmissions. However, integrating the scheme proposed therein with cooperative communications would require the implementation of complex coding schemes among all the cooperating nodes in the system. Other works, such as \cite{Larsson}, discuss the maximization of the throughput for multiple input multiple output (MIMO) ARQ systems, without considering the power allocation across different ARQ transmission rounds. However, in URLLC systems the goal is not necessarily the design of high throughput systems, but the design of robust and reliable systems, while maintaining reasonably low complexity algorithms at the devices, and thus low energy consumption. In this context, in \cite{Dosti} the authors propose closed form OPA scheme for ARQ protocol that enables communication with ultra-reliability constrains. Therein, the authors show that the proposed scheme maximizes the overall system throughput with or without the presence of feedback delay. In \cite{Dosti2}, the authors propose an OPA scheme for CC-HARQ protocol, which allows the exploitation of the coding gains across the collected packets at the receiver. As expected, this results in larger power savings. However, their analysis is limited only to point-to-point communications.
 
In this paper, we integrate the AF cooperative scheme in ARQ retransmission scheme in a two-hop single-relay network, in order to achieve the same performance in terms of outage probability, but with lower power expenditure while maintaining low complexity \cite{LTS}. More specifically, 
we develop an AF-ARQ relay scheme that enables communication in the wireless channel with minimum power expenditure while guaranteeing a target reliability level for both symmetric and asymmetric power allocation between source and relay. 
The contributions of this paper are as follows:
\begin{list4}
\item[1)] We obtain a closed form expression for the outage probability of AF-ARQ relay scheme under asymmetric power allocation between the source and the relay.
\item [2)] We obtain a closed form expression for the power allocation strategy that minimizes the  outage-weighted  average transmitted power.
\item [3)]  Through simulations, we show that the proposed scheme provides a better performance than the EPA strategy and different types of point-to-point communication schemes.
\end{list4} 

	
%
%
%
%
\section{System model}
\label{sec:system_model}

In this work, we consider a two-hop relay network consisting of a source $S$, a relay $R$ and a destination $D$. The source follows an ARQ protocol combined with AF relaying strategy. Each ARQ round consists of $1$ time slot of duration $T$, and this slot is divided into two phases: 
\begin{list4}
\item[1)] \label{(1)} In the \emph{first phase}, the source broadcasts the packet to the destination and relay. 
\item [2)]  In the \emph{second phase}, the relay amplifies the received signal and forwards it to the destination. If the packet is correctly recovered by the destination an acknowledgment packet (ACK) is fed back and the source carries on with transmitting the next packet. Otherwise, the destination sends a negative-acknowledgment packet (NACK) and a new ARQ round is initiated starting from the first phase.
\end{list4} 

The two phases described above are carried for a maximum number of $M$ ARQ rounds. If all these rounds are unsuccessful, then a failure to transmit the packet is declared, and the source proceeds to sending a new packet. Within one ARQ round, we consider the simplest orthogonal separation between terminals in two-phase time-division multiplexing (TDM). In the first $T/2$ of the slot the source broadcasts the message with power $P_s$ to both destination and the relay, whose respective received messages would be:
\begin{align}
y_{r} = \sqrt{P_s}h_{sr}x_{s} + n_r,  \label{eq:y_r}  \\
y_{d} = \sqrt{P_s}h_{sd}x_{s} + n_{d},  \label{eq:y_d}
\end{align}
where $y_r$, $y_{d}$ denote the received messages by the relay and the destination, respectively; $h_{sr}$, $h_{sd}$ are the source-relay and source-destination channel coefficients, which are statistically independent complex normal random variables with mean zero, and variance $\sigma_{ij}^2$, i.e, $h_{ij} \sim \mathcal{CN} (0, \sigma_{ij})$; the envelope of the channel coefficients is Rayleigh distributed, i.e., $|h_{ij}|\sim \rayleigh (\sigma_{ij})$. The channel gains $g_{ij} \triangleq |h_{ij}|^2$ are, therefore, exponentially distributed, i.e., $g_{ij} \sim \exponential (\sigma^{-2}_{ij}/2)$.  Terms $n_r$ and $n_{d}$ represent the additive white Gaussian noise (AWGN) at the relay and the destination, respectively, both with power $N_r$ and $N_d$. Without loss of generality we assume $N_r = N_d = 1$. 
	
In the remaining $T/2$ of the slot, the relay amplifies the received signal by a gain factor (for more details see \cite{LTS}), and forwards it to the destination, which receives: 
\begin{align}
\label{2.2}
y_{d} = \sqrt{\frac{P_r}{h_{sr} + n_r}}h_{rd}y_{r} + n_{d},
\end{align}
where $x_{r}$ is the signal transmitted by the relay and $P_r$ is the power of transmission of the relay and $n_{rd}$ is the noise associated with the relay-destination channel. At the end of the two slots, the destination performs maximum ratio combining (MRC) of the received copies of the packet to recover the information. To perform MRC, we assume that the receiving terminals can estimate the channel gain coefficients with high accuracy. We assume that the transmitter knows only the distribution of the channel coefficients.

Motivated by URLLC applications, where short packets have to be sent with high reliability, we assume that we have quasi-static fading channel conditions among ARQ rounds, i.e., the channel coefficients $h_{ij}$ from transmitter $i$ to receiver $j$ remain constant for the duration of one time slot of duration $T$ and change independently between ARQ rounds. Lastly, we assume to have one-bit feedback, which is instantaneous and error free. 
	
%
%
%
%
%
%
\section{Outage Behavior}
\label{III}
In this section, we compute the outage probability formula for the transmission within one ARQ round. The derivations in this section follow closely the work presented earlier in \cite{LTS}. However, the results presented there are limited only to the case when the source and the relay powers are the same, i.e, $P_s= P_r = P$. 

During one ARQ transmission round, the mutual information accumulated from the destination terminal is \cite{LTS, Laneman}
\begin{align}
\label{3.1}
I_{AF} = \frac{1}{2} \log\left(1 + P_s|h_{sd}|^2 + f(P_s|h_{sr}|^2, P_r|h_{rd}|^2)\right),
\end{align}
where $f(x,y) = \frac{xy}{x+y+1}$. For a certain spectral efficiency $R = p/q$, where $p$ is the number of information bits and $q$ is the number of channel uses, an outage occurs when $I_{AF} < R$\footnote{
	Notice that hereafter in order to standardize the notation we assume that all information is encoded in nats instead of bits. Therefore, $\log$ is the natural logarithm.
}. 
\begin{theorem}
	For asymmetric power allocation, when $P_s$ is variable, and the ratio $\frac{P_s}{P_r} \leq \frac{\mu}{\delta} <\infty$, the outage probability of one ARQ round can be found as:
	\begin{align}
	\label{3.3}
	\epsilon = \frac{P_s \sigma_{sr}^2 + P_r \sigma_{rd}^2}{2 P_r \sigma_{sd}^2 \sigma_{sr}^2 \sigma_{rd}^2} \left(\frac{e^{2R}-1}{P_s}\right)^2 .
	\end{align}
\end{theorem}
\begin{proof}
See Appendix~\ref{A}.
\end{proof}
In what follows, we compute the probability that the packet is not decoded correctly even after a certain number of transmission rounds, $m$, has occurred. This follows from the assumption that all the transmissions of the packets experience independent fading conditions, so the total outage probability becomes
\begin{align}
\label{eq:16}
E_M=\prod_{m=1}^M \epsilon_m \mathrm{,}
\end{align}
where $\epsilon_m$ is the outage probability of the $m^{th}$ ARQ round. Since no transmission is done at round $m=0$, the outage probability is $\epsilon_0=1$.  

%
%
%
%
\section {Optimal power allocation strategy}
\label{IV}
In this section, we provide the OPA strategy across different ARQ rounds in closed form. The problem of interest is to achieve a target outage probability while spending as little power as possible for sending the information from the transmitter to the receiver. Naturally, since no Channel State Information (CSI) is available at the transmitter, one approach would be to allocate the same amount of power across all transmission rounds. We show, in the numerical evaluation, that our proposed power allocation scheme has lower power expenditure for a fixed target outage probability. 

The outage-weighted average transmitted power is defined as
\begin{align}
\label{eq:15}
P_{\mathrm{avg}} :=\frac{1}{M}\sum_{m=1}^M P_m E_{m-1}\mathrm{,}
\end{align}
where $M$ is the maximum number of ARQ rounds, $P_m$ is the power transmitted in the $m^{th}$ round and $E_{m-1}$ is the outage probability up to round $m-1$.	Mathematically, the problem of interest can be formulated as follows
\begin{equation}
\begin{aligned}
\label {eq:17}
& {\text{minimize}}
& &P_{ \mathrm{avg}} \\
& \text{subject to}
& & 0 \leq P_{m}, \; m \in \{1, \ldots, M\} \\
&&& E_M = \epsilon ,
\end{aligned}
\end{equation}
where $\epsilon$ is any target outage probability. Problem \eqref{eq:17} is a Geometric Program (GP) and in what follows we will derive its closed form solution, which allows for much faster computation of the problem solution when compared to running an optimization algorithm. This is essential in the context of URLLC, since it results in minimization of the end-to-end delay. Furthermore, in the numerical section we utilize the CVX GP solver and show that the proposed OPA matches the solver's results. Eq.~\eqref{3.3} at ARQ round $m$ can be re-written as
\begin{align}
\label{new_outage}
\epsilon_m = \psi(\eta_m)\left(\frac{\phi_m}{P_{m}}\right)^2, \; m \in \{1, \ldots, M\},
\end{align}
where $\eta_m$ can be computed as the ratio of the power transmitted by the relay in the $m^{th}$ round to the power transmitted by the source in the $m^{th}$ round, i.e., $\eta_m := P_{r,m}/P_m$. Furthermore, $\psi(\eta_m) := \frac{\frac{1}{\eta_m} \sigma_{sr}^2 + \sigma_{rd}^2}{2  \sigma_{sd}^2 \sigma_{sr}^2 \sigma_{rd}^2}$ and $\phi_m(R) = e^{2R/m} -1$. In this work, we assume that the amplifier's gain factor $\alpha$ is set \emph{a priori}, i.e., $\eta_{m}$ is fixed, and we only optimize over the source power. The OPA strategy for optimization problem~\eqref{eq:17} is given in Theorem~\ref{theorem:2}.

\begin{theorem}\label{theorem:2}
	The OPA strategy for the AF-ARQ protocol is as follows
	\begin{align}
	\label{eq,c1:P_M}
	P_M&=\sqrt[3]{ 2 \lambda \phi_M \psi(\eta_M)},
	\\ \label{rec}
	P_m&=\sqrt{3 \phi_m \psi(\eta_{m}) P_{m+1}}, \; i \in \{1, \ldots, M-1\}.
	\end{align}
\end{theorem}
\begin{proof}
See Appendix~\ref{C}.
\end{proof}
\begin{remark} \label{remark1}
	In this work, we have limited our analysis only to the case where we need to optimize over the source, while allowing for suboptimal selection of $\eta_m$. From the resulting closed form solution, it is easily observed that all the power terms can be computed recursively, i.e., for a fixed number of transmissions, initial spectral efficiency and $\eta_M$, the power term in the $M^{th}$ round can be prespecified; the other power terms are then recursively computed as described in \eqref{rec}. 
\end{remark} 	
\section{Numerical section}
\label{V}
In this section, we evaluate the performance of the proposed power allocation scheme for different values of $M$ and $\eta$. We also provide comparison with the EPA strategy and the point-to-point optimized ARQ scheme \cite{Dosti}. In what follows, we assume that the initial spectral efficiency $R = 1$ nat per channel use (npcu)and the channel statistics are $\sigma_{sd}^2 = 2$ and $\sigma_{sr}^2 = \sigma_{rd}^2 = 1$.

In Fig.~\ref{fig1}, we illustrate the power that we should allocate in each AF-ARQ round to achieve a target outage probability $\epsilon$ for $M=2$ and different values of $\eta$. First, notice  that by allowing for different power allocation between the source and the relay provides some performance gain (in the sense lower power expenditure) when compared to the case when $P_s = P_r = P$, for both source power terms $P_1$ and $P_2$ while guaranteeing the required reliability level of the system. However, note that allowing the relay to transmit with arbitrarily high power, i.e., increasing $\eta$, does not provide large power savings with respect to the source. This suggests that for a fixed total power budget $P_{\mathrm{total}}$ (i.e., $P_s+P_r \leq P_{\mathrm{total}}$), then there exists an optimal value of $\eta$. As stated in Remark~\ref{remark1}, this problem falls outside the scope of this paper and will be further investigated in later works. Moreover, notice that the OPA strategy suggests transmission with increasing power across subsequent AF-ARQ rounds when lower outage probability values are required in the system. This result is consistent with the ones attained earlier in \cite{Dosti, Dosti2} where point-to-point ARQ systems are analyzed. The intuition behind the result follows from the way the optimization problem is formulated. To achieve a certain target outage probability for $M = 2$ while minimizing power expenditure for a fixed latency, it is better to first transmit with low power and ``hope'' to encounter good channel conditions. In case a failure occurs in the first round, then insist again with a higher power transmission until success or exhaustion of the maximum number of allowed transmissions.
\begin{figure}[!htb] 
	\centering
	\includegraphics[trim=3.7cm 9cm 4cm 10cm, clip=true,height = .81\linewidth, width=1.0\columnwidth]{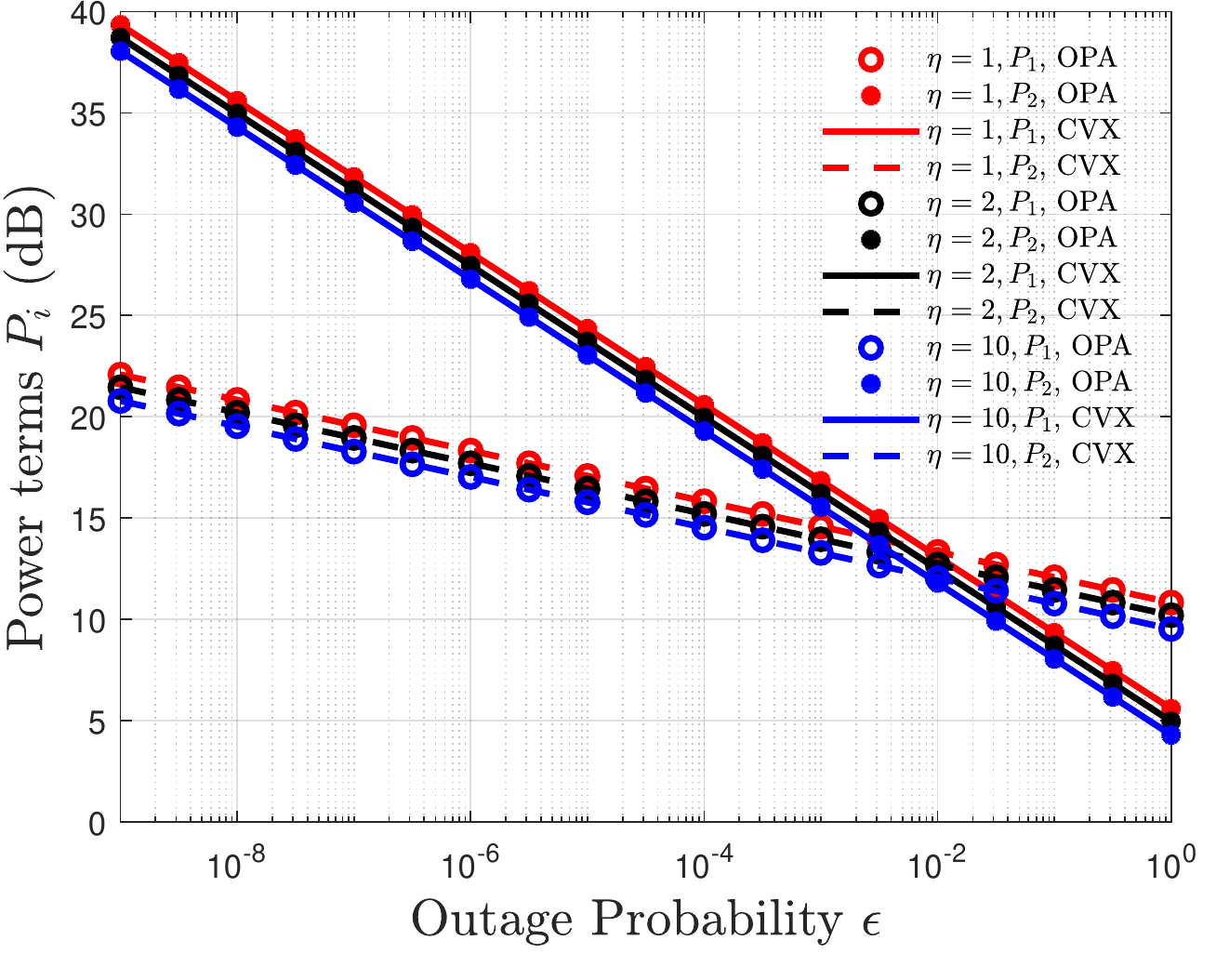}
	\caption{Optimized transmit power in each AF-ARQ round for $M=2$, $R=1$ npcu and different values of $\eta$. }
	\label{fig1}
\end{figure} 

Next, in Fig.~\ref{fig2} we compare the performance of our protocol with respect to the optimized ARQ protocol, derived in \cite{Dosti}. In the plots, we fix $M=3$ and $\eta = 1$. First, we observe that cooperation provides large power savings at the source with respect to point-to-point ARQ. Intuitively, this follows from the fact that since the link between the source and destination has high variance, in the AF-ARQ case the relay, which (usually) has a better link to the destination, can help. On the other hand, in the point-to-point case, the source has to insist with high power in the latter rounds, which yields much higher power consumption. Obviously, in the case of point-to-point ARQ the average power consumption would be larger than AF-ARQ. Furthermore, the presence of the relay gives more robustness to the network in the cases when deep fades are present, which is essential for low-latency communication systems. Another interesting observation follows from the behavior of the power terms for AF-ARQ. For very low outage probability values, e.g., $\epsilon = 10^{-9}$, the optimal allocation strategy is still transmission with incremental power. However, for less stringent requirements on the outage probability the optimal transmission behavior changes, which suggests that there is a trade-off between the delay and power minimization. Intuitively, this happens because moderate outage probabilities are ``easier'' to achieve for ARQ-type schemes. Therefore, the OPA strategy suggests transmission with high power at first, which enables delay minimization. However, mathematical characterization of the change in OPA behavior for different values of $M$ still needs to be understood better and remains an open problem.
\begin{figure}[!htb] 
	\centering
	\includegraphics[trim=3.5cm 9cm 4cm 10cm, clip=true,height = .81\linewidth, width=1.0\columnwidth]{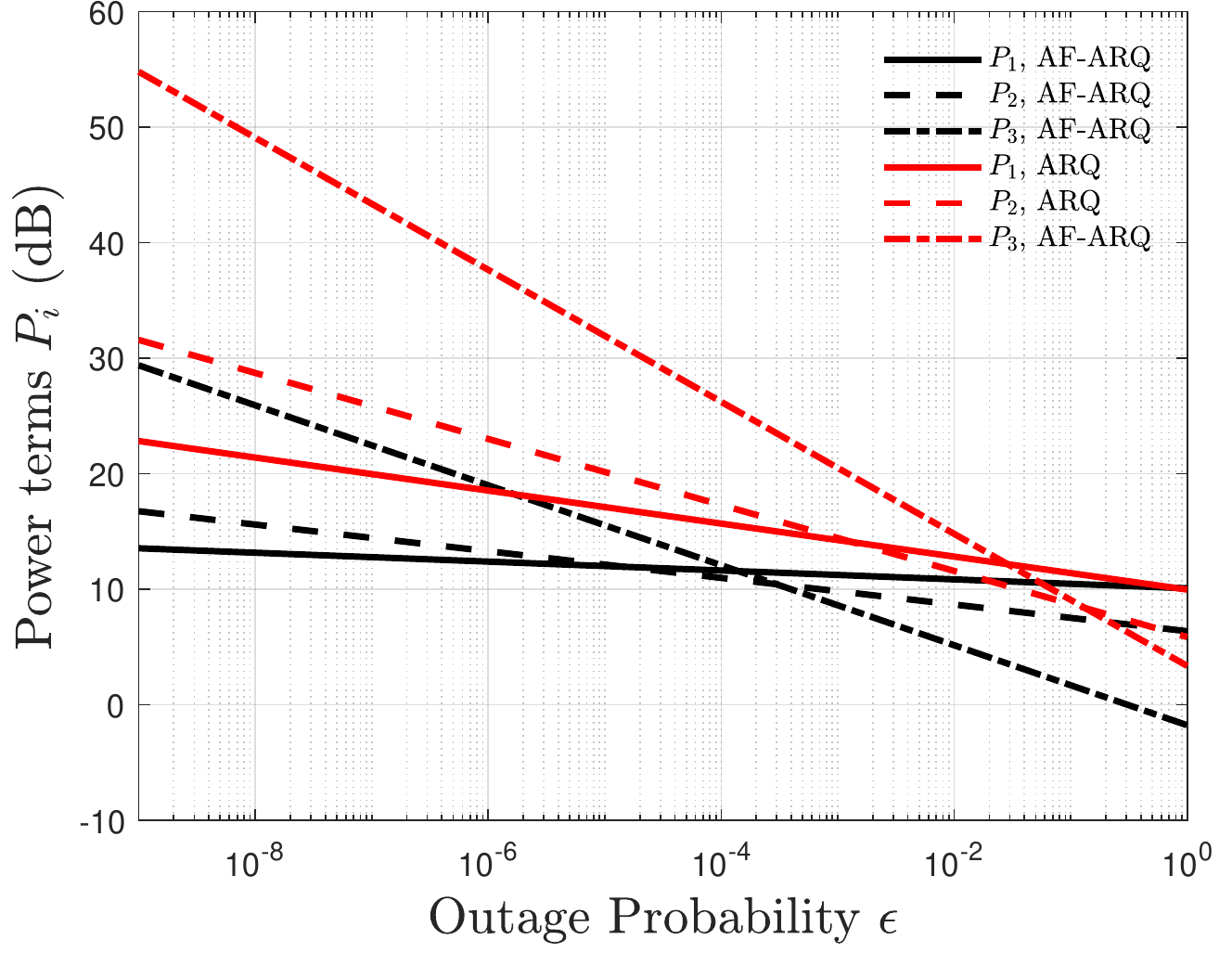}
	\caption{Comparison between optimized transmitted power in AF-ARQ and ARQ $M = 3$, $\eta=1$, $\sigma_{sd}^2 = 2$, $\sigma_{sr}^2 = \sigma_{rd}^2 = 1$ and $R = 1$ npcu.}	
	\label{fig2}
\end{figure}

Lastly, in Fig.~\ref{fig3}, we evaluate the outage-weighted average power expenditure per transmission of our protocols in the case of $M \in \{2, 3\}$ transmissions. Herein, we compare the performance of our scheme with the EPA across different AF-ARQ rounds. To obtain the latter, we can substitute $P_i, \forall i\in\{1 \ldots M\}$ in \eqref{eq:16}. Notice that for conventional outage probabilities, i.e., $\epsilon \in [10^{-2}, 10^{-4}]$, the OPA strategy makes little difference. This result is already presented in \cite{Makki}. However, for very tight reliability constraints, the EPA approach is strictly suboptimal.

\begin{figure}[!htb] 
	\centering
	\includegraphics[trim=3.5cm 9cm 4cm 10cm, clip=true,height = .81\linewidth, width=1.0\columnwidth]{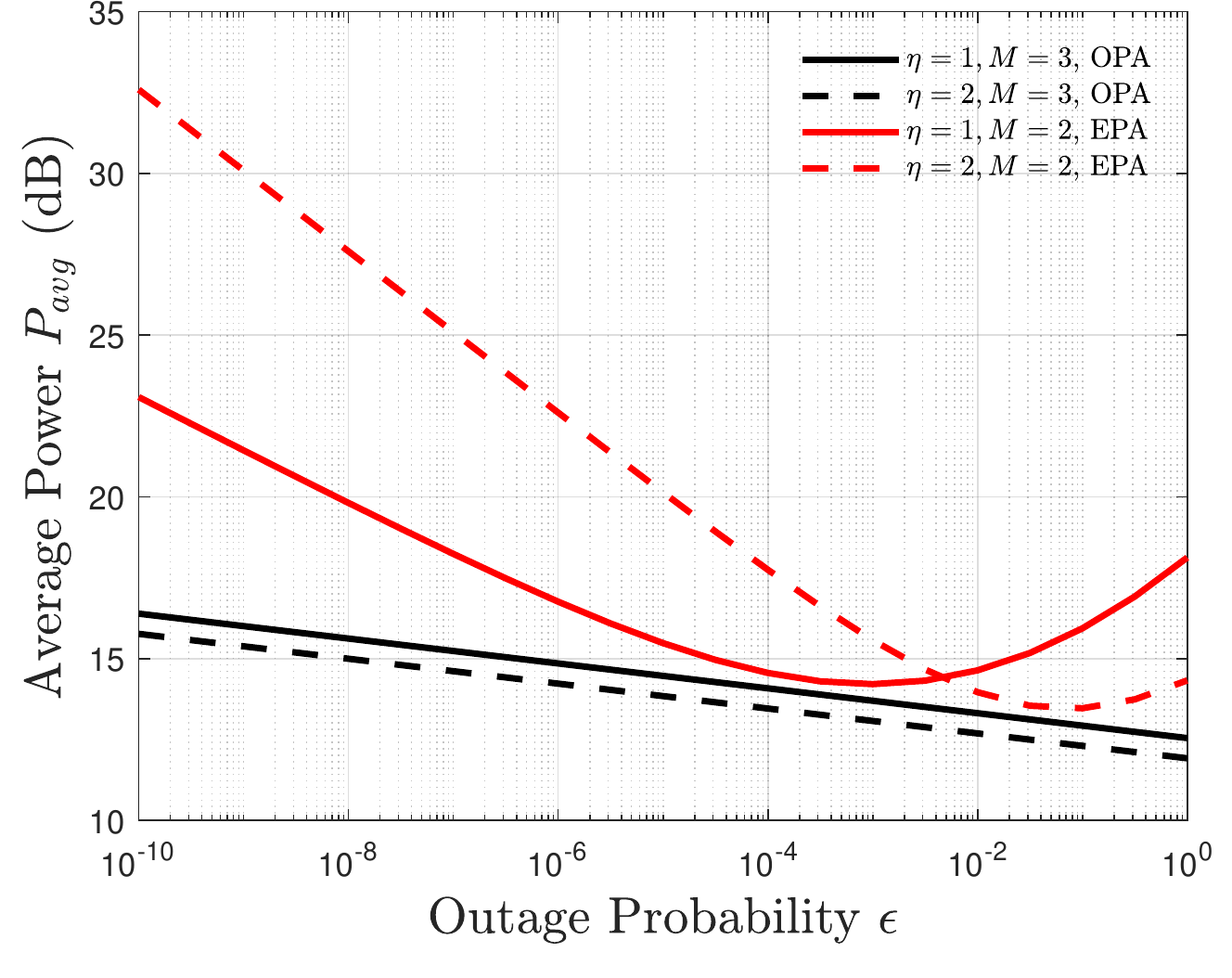}
	\caption{Comparison between OPA and EPA in AF-ARQ for $M \in \{2, 3\}$, $\eta \in \{1, 2\}$ and $R = 1$ npcu.} 
	\label{fig3}
\end{figure}

\section{Conclusions and Future Directions}

In this paper, we analyzed the implementation of AF-ARQ protocol for URLLC. First, we obtained a closed form expression for the outage probability of the protocol. Next, we computed the OPA scheme among different transmission rounds for the protocol. The proposed scheme allows operation in very low outage probabilities while minimizing the average outage-weighted power expenditure. We showed that the proposed strategy suggests transmission with increasing power in each ARQ round. Furthermore, through simulations we showed that the proposed scheme outperforms the selected benchmarks and produces large power gains. 


There are many extensions to this work. First, it would be of high interest to provide a mathematical characterization of the trade-off between delay and power minimization. Secondly, we believe that it would be interesting to evaluate the OPA strategy under different channel models, such as, Ricean or Nakagami-m. Another insightful extension would follow from the analysis the scenario when both $P_s$ and $P_r$ are allowed to be free variables. Moreover, it would be interesting to see how extending the number of relays in the network would help with improving the overall system performance. Lastly, we believe that it would be interesting to look at the problem of minimizing the end-to-end latency for a fixed power budget and target outage probability. 

\begin{appendices}
	\section{Proof of Theorem 1.}
	\label{A}
	In the following, for notation convenience we denote $X\triangleq |h_{sd}|^2$, $Y\triangleq |h_{sr}|^2$, $Z\triangleq |h_{rd}|^2$, $P_s = 1/\delta$ and $P_r = 1/\mu$. To prove the theorem, we need the following lemmas:
	\begin{lemma}
		Let $g(x)$ be a continuous function around some point $x=x_0$, satisfying $g(x) \rightarrow 0$ as $x \rightarrow x_0$. If $X~\sim \exponential (\lambda_X)$ we have: 
		\begin{align}
		\label{p.f} 		
		\lim_{x \rightarrow x_0} \frac{1}{g(x)} \mathbb{P} \left(X <g(x)\right) = \lambda_{X},
		\end{align}
		where $\mathbb{P} \left(X <g(x)\right)$ denotes the probability of random variable (RV) $X$ to be less than $g(x)$.
	\end{lemma}
	
	\begin{proof}
		\begin{align*}
		\lim_{x \rightarrow x_0}  \frac{1}{g(x)} \mathbb{P} \left(X <g(x)\right) &\stackrel{(a)}{=}  \lim_{x \rightarrow x_0} \frac{1-e^{-\lambda_{X} g(x)}}{g(x)} 
		\\
		&\stackrel{(b)}{=}  \lim_{x \rightarrow x_0} \frac{\lambda_{X} g'(x) e^{-\lambda_{X} g(x)}}{g'(x)} \stackrel{(c)}{=} \lambda_{X},
		\end{align*}
		where $(a)$ follows from the definition of the Cumulative Distribution Function (CDF) of exponential distribution; $(b)$ follows from L'Hopitals rule; $(c)$ follows by noting that the $\lim_{x \rightarrow x_0} e^{-\lambda_{X} g(x)} =1$.
	\end{proof}
	
	\begin{lemma}
		Let $(\mu, \delta) > (0, 0)$ and $r_{\mu, \delta} = \delta f(\frac{Y}{\delta}, \frac{Z}{\mu})$, where $f(x,y) = \frac{xy}{x+y+1}$. Let $h ( \delta) > 0$ be a continuous function with $h ( \delta) \rightarrow 0$ and $\frac{\delta}{h ( \delta)} \rightarrow d_{1} < \infty$ as $\delta \rightarrow 0$. Let also $\frac{\mu}{\delta} = a < \infty$ as $\delta \rightarrow 0$. Then,
		\begin{align}
		\lim_{\delta \rightarrow 0} \frac{1}{h ( \delta)} \mathbb{P} \left( r_{\mu, \delta} <h (\delta) \right) = \lambda_Y + \frac{\mu}{\delta}\lambda_Z.
		\end{align}
	\end{lemma}
	
	\begin{proof} 
		We start with the lower bound:
		\begin{align}
		\nonumber
		&\mathbb{P} \left( \delta f\left(\frac{Y}{\delta}, \frac{Z}{\mu}\right) < h (\delta) \! \right) \! \\ \nonumber  
		&\stackrel{(a)}{=} \mathbb{P} \left( \frac{\frac{Y}{\delta}+ \frac{Z}{\mu}+1}{Y \frac{Z}{\mu}} \geq \frac{1}{h (\delta)}\right)
		\\ 
		\nonumber 
		&\stackrel{}{=} \mathbb{P} \left(\frac{1}{\frac{Z \delta}{\mu}} + \frac{1}{Y} + \frac{\mu}{YZ} \geq \frac{1}{h (\delta)}\right)
		\\ \nonumber
		&\stackrel{(b)}{\geq} \mathbb{P} \left(\frac{1}{\frac{Z \delta}{\mu}} + \frac{1}{Y} \geq \frac{1}{h (\delta)}\right)
		\\ \nonumber
		&\stackrel{(c)}{\geq} \mathbb{P} \left( \max \left( \frac{1}{\frac{Z \delta}{\mu}}, \frac{1}{Y}\right) \geq \frac{1}{h (\delta)}\right)
		\\ \nonumber
		&\stackrel{}{=} 1- \mathbb{P} \left(\frac{1}{Y} \geq \frac{1}{h (\delta)}\right) \mathbb{P} \left(\frac{1}{\frac{\delta}{\mu} Z} \geq \frac{1}{h (\delta)} \right)
		\\ \nonumber
		&\stackrel{}{=} 1- \left(\int_{0}^{h_{1} (\delta)} \lambda_Y e^{-\lambda_{Y} y} dy\right) \left(\int_{0}^{h_{1} (\delta)} \frac{\mu}{\delta}\lambda_Z e^{-\frac{\mu}{\delta} \lambda_{Z} z} dz\right)
		\\ 
		&\stackrel{}{=} 1- e^{-\left(\lambda_Y + \frac{\mu}{\delta}\lambda_Z\right) h (\delta)},
		\end{align}
		where $(a)$ follows by making some algebraic manipulations; $(b)$ follows from the fact that we are discarding a positive term, thus reducing the chances of the event happening; $(c)$ follows again from the fact that a positive term is discarded. Since $\delta$ and $\mu$ denote the powers of the source and the relay, respectively, it is plausible to assume that $\frac{\mu}{\delta} = a < \infty$. Therefore, 
		\begin{align}
		\label{p.5}
		\inf \lim_{\delta \rightarrow 0} \frac{1}{h (\delta)} \mathbb{P} \left( \delta f\left(\frac{Y}{\delta}, \frac{Z}{\mu}\right) < h (\delta) \right)  \geq \lambda_Y + a \lambda_Z.
		\end{align}
		Next, we prove the converse. Let $l >1$ be a constant. Consider:
		\begin{align}
		\nonumber
		&\mathbb{P}\! \left(\! \delta f\left(\frac{Y}{\delta},  \frac{Z}{\mu}\right) \! < \! h (\delta) \! \right) \! \\
		\nonumber
		&\stackrel{(a)}{=} \mathbb{P} \left( \frac{1}{Y} \left(1 + \frac{\mu}{Z}\right) \geq \frac{1}{h (\delta)} - \frac{1}{\frac{Z}{a}}\right)
		\\ \nonumber
		&\stackrel{(b)}{=} \int_{0}^{\infty} \mathbb{P} \left( \frac{1}{Y} \geq \frac{\frac{1}{h (\delta)} - \frac{1}{\frac{Z}{a}}}{1 + \frac{\mu}{Z}}\right) p_{\frac{Z}{a}} \left(\frac{z}{a}\right) d\frac{z}{a}
		\\ \label{p.7}
		&\stackrel{(c)}{=} \underbrace{\int_{0}^{l h (\delta)} \mathbb{P} \left( \frac{1}{Y} \geq \frac{\frac{1}{h (\delta)} - \frac{1}{\frac{Z}{a}}}{1 + \frac{\mu}{Z}}\right) p_{\frac{Z}{a}} \left(\frac{z}{a}\right) d\frac{z}{a}}_{\triangleq B_1} \\ \nonumber &+ \underbrace{\int_{l h (\delta)}^{\infty} \mathbb{P} \left(  \frac{1}{Y} \geq \frac{\frac{1}{h (\delta)} - \frac{1}{\frac{Z}{a}}}{1 + \frac{\mu}{Z}}\right) p_{\frac{Z}{a}} \left(\frac{z}{a}\right) d\frac{z}{a}}_{\triangleq B_2} ,
		\end{align}
		where $(a)$ stems after algebraic manipulation. The integrals in $(b)$ and $(c)$ are over the exponential RV $\frac{Z}{a}$ and $p_{\frac{Z}{a}} \left(\frac{z}{a} \right)$ denotes the pdf of that random variable. Next, we start to bound each of the terms in \eqref{p.7}. To bound $B_1$ we start by noticing that \vspace{-1mm}
		\begin{align}
		\nonumber
		B_1 &\stackrel{(a)}{=} a\lambda_{Z} \int_{0}^{l h (\delta)} \mathbb{P} \left( \frac{1}{Y} \geq \frac{\frac{1}{h (\delta)} - \frac{1}{\frac{Z}{a}}}{1 + \frac{\mu}{Z}}\right) e^{-a\lambda_{Z} z} d\frac{z}{a}
		\\
		\label{p.9}
		&\stackrel{(b)}{\leq} a \lambda_{Z} l,
		\end{align}
		where $(a)$ follows from the definition of the pdf of the RV $\frac{Z}{a}$; $(b)$ follows from the facts that the integral in $(a)$ produces a number smaller than or equal to one and $l > 1$.
		
		Next, to bound $B_2$ we let $k > l > 1$ be another constant. Further, let \vspace{-3.5 mm}
		\begin{align}
		\label{p.10}
		&B_2 =\underbrace{ \int_{l h (\delta)}^{k h (\delta)} \mathbb{P} \left( \frac{1}{Y} \geq \frac{\frac{1}{h (\delta)} - \frac{1}{\frac{Z}{a}}}{1 + \frac{\mu}{Z}}\right) p_{\frac{Z}{a}} \left(\frac{z}{a}\right) d\frac{z}{a}}_{\triangleq B_{21}} \\ \nonumber &+ \underbrace{\int_{k h (\delta)}^{\infty} \mathbb{P} \left(  \frac{1}{Y} \geq \frac{\frac{1}{h (\delta)} - \frac{1}{\frac{Z}{a}}}{1 + \frac{\mu}{Z}}\right) p_{\frac{Z}{a}} \left(\frac{z}{a}\right) d\frac{z}{a}}_{\triangleq B_{22}} . 
		\end{align}
		We start by bounding $B_{22}$:
		\begin{align}
		\nonumber B_{22} &\stackrel{(a)}{=} \int_{k h (\delta)}^{\infty} \mathbb{P} \left(  \frac{1}{Y} \geq \frac{Z - a h (\delta)}{\left(Z + \mu\right) h (\delta)}\right) p_{\frac{Z}{a}} \left(\frac{z}{a}\right) d\frac{z}{a}
		\\ \nonumber
		&\stackrel{(b)}{\leq} \mathbb{P} \left(\frac{1}{Y} \geq \frac{k h (\delta) - a h (\delta)}{\left(kh (\delta) + \mu\right)h (\delta)}\right)
		\\ \nonumber
		&\stackrel{}{=} \mathbb{P} \left(\frac{1}{Y} \geq \frac{1- \frac{a}{k}}{h (\delta) + \frac{\mu}{k}}\right)
		\\ 
		&\stackrel{}{=} \mathbb{P} \left(Y \leq \frac{h (\delta) + \frac{\mu}{k}}{1- \frac{a}{k}}\right) \label{d},
		\end{align}
		where $(a)$ is obtained by making some algebraic manipulations to the original expression in \eqref{p.10}; $(b)$ follows because the argument of $\mathbb{P}$ is non-increasing in $Z/a$. Since \eqref{d} is the CDF of the random variable $Y$, we can utilize the result of Lemma 1. Consider
		\begin{align}
		\mathbb{P} \left(Y \leq \frac{h (\delta) + \frac{\mu}{k}}{1- \frac{a}{k}}\right)\! &\stackrel{(a)}{=}\! \mathbb{P} \left( \! Y \leq \frac{ 1 + \frac{\mu}{k h (\delta)}}{1- \frac{a}{k}} h (\delta) \right)
		\\
		\label{p.12}
		&\stackrel{(b)}{\leq} \lambda_{Y} \frac{ 1 + \frac{\mu}{k h (\delta)}}{1- \frac{a}{k}} h (\delta), 
		\end{align} 
		where the bound in $(b)$ is obtained via a similar argument as in \eqref{p.9}. 
		
		Finally, we bound the remaining term, $B_{21}$: 
		\begin{align}
		\nonumber
		&B_{21} \stackrel{(a)}{=} \int_{l}^{k} \mathbb{P} \left(\frac{1}{Y} \geq \frac{\frac{1}{h (\delta)} - \frac{1}{Z' h (\delta)}}{1 + \frac{\mu}{Z' a h (\delta)}}\right) \\ 
		\nonumber &\quad \quad \quad \lambda_{Z' h (\delta)} e^{-\lambda_{Z' h (\delta)} z' h (\delta)} h (\delta) dz'
		\\ \nonumber
		&\stackrel{(b)}{\leq} \lambda_{Z' h (\delta)} \int_{l}^{k} \mathbb{P} \left(\frac{1}{Y} \geq \frac{\frac{1}{h (\delta)} - \frac{1}{Z' h (\delta)}}{1 + \frac{\mu}{Z' a h (\delta)}}\right) h (\delta) dz'
		\\ \nonumber
		&\stackrel{}{\leq}\! \lambda_{Z' h (\delta)} h_2^2 (\delta) \!\int_{l}^{k}\! \frac{1}{h (\delta) } \! \left(\!  \mathbb{P}\! \left(\! Y\! < \! \frac{1 + \frac{\mu}{Z' a h (\delta)}}{\frac{1}{h (\delta)} - \frac{1}{Z' h (\delta)}} \!\right)\! \right) \! dz'
		\\
		\label{p.11}
		&\stackrel{(d)}{=} \lambda_{Z' h (\delta)} h_2^2 (\delta) \int_{l}^{k} \lambda_{Y} \frac{1 + \frac{\mu}{Z' a h (\delta)}}{\frac{1}{h (\delta)} - \frac{1}{Z' h (\delta)}}dz' \\ \nonumber& = h_2^2 (\delta) \gamma (h (\delta), \delta, l, k),
		\end{align}
		where $(a)$ follows from making the change of RV $Z' = \frac{Z}{a h (\delta)}$ and applying the definition of $p_{Z'} (z')$; $(b)$ follows from the fact the exponential function is decaying in $Z'$; $(d)$ would follow from \eqref{p.f}, and $\gamma (h (\delta), \delta, l, k)$ is finite for all $k > l > 1$ as $\delta \rightarrow 0$. 
		
		Finally, combining \eqref{p.9}, \eqref{p.12}, \eqref{p.11} we obtain:
		\begin{align*}
		\sup \lim_{\delta \rightarrow 0} \! \frac{1}{h (\delta)} \mathbb{P}\! \left(r_{\mu, \delta}\! <\! h (\delta)\! \right) &\stackrel{}{\leq} \sup \lim_ {\delta \rightarrow 0} \frac{1}{h (\delta)} \left( B_1\! + \!B_{21} \!+\! B_{22} \right)
		\\
		&\stackrel{}{=} a \lambda_{Z} l + \lambda_{Y}
		\\
		&\stackrel{(a)}{=} a \lambda_{Z}+ \lambda_{Y}, \vspace{-3mm}
		\end{align*}
		where the remainder terms vanish as $\delta \rightarrow 0$ and $k$ is selected large enough; $(a)$ follows from choosing $l \approx 1$. 
	\end{proof}
	
	\begin{proposition} 
		Let $X \sim \text{exp} (\lambda_{X})$, $\delta > 0$ and $g(\delta) > 0$ be continuous with $g(\delta) \rightarrow 0$ and $\delta/g(\delta) \rightarrow c < \infty$ as $\delta \rightarrow 0$. Then:
		\begin{align}
		\lim_{\delta \rightarrow 0} \frac{1}{g^2(\delta)}\mathbb{P} \left(X + r_{\mu, \delta} < g(\delta)\right) = \frac{\lambda_{X} (\lambda_Y + c \lambda_Z)}{2}.
		\end{align}
	\end{proposition}
	\begin{proof}
		Building upon the previous lemma, we now consider 
		\begin{align}
		\nonumber
		&\mathbb{P} \left(X + r_{\mu, \delta} < g(\delta) \right) \! \stackrel{(a)}{=} \! \int_{0}^{g(\delta)} \!\!\! \! \mathbb{P} \left(r_{\mu, \delta} < g(\delta) - X \right) p_X(x) dx
		\\ \nonumber
		&\stackrel{}{=} g(\delta) \int_{0}^{g(\delta)} \mathbb{P} \left(r_{\mu, \delta} < g(\delta) \left(1 - \frac{X}{g(\delta)}\right)\right) \lambda_{X} e^{-\lambda_{X} x} dx
		\\ 
		&\stackrel{(b)}{=}\! \! \label{c.15} g(\delta) \! \int_{0}^{1}\! \! \! g(\delta)\! \left( \! 1 - x' \!\right)\! \frac{\mathbb{P} \! \left(r_{\mu, \delta}\! <\! g(\delta) \! \left(1 \! - \! X' \! \right)\!\right)}{g(\delta)\! \left(1\! -\! x'\right)} \lambda_{\!X\!} e^{-\lambda_{\!X\!} g(\delta) x'}\! \! dx' \! \!,
		\end{align}
		where $(a)$ comes from marginalizing out $X$ (notice that $r_{\mu, \delta} \geq 0$ when $X \in \left[0, g(\delta)\right]$); $(b)$ is obtained by the change of variable $X' = \frac{X}{g(\delta)}$. To complete the proof, now we consider:
		\begin{align}
		&\lim_{\delta \rightarrow 0} \frac{1}{g^2(\delta)} g(\delta) \int_{0}^{1} g(\delta) \left(1 - x' \right) \frac{\mathbb{P} \left(r_{\mu, \delta}< g(\delta) \left(1 - X' \right)\right)}{g(\delta) \left(1 -x'\right)} \\ &\nonumber \lambda_{X} e^{-\lambda_{X} g(\delta) x'}dx
		\\ \nonumber
		&\stackrel{(a)}{=} \int_{0}^{1} \left[\lim_{\delta \rightarrow 0} \frac{\mathbb{P} \left(r_{\mu, \delta} < g(\delta) \left(1 - x'\right)\right)}{g(\delta) \left(1 - x'\right)} \right] \lambda_{X} \\ \nonumber &\lim_{\delta \rightarrow 0} e^{-\lambda_{X} g(\delta) x'} dx',
		\\ \nonumber
		&\stackrel{(b)}{=} \int_{0}^{1} \left(1 - x'\right) \left(\lambda_Y + c \lambda_Z\right) \lambda_{X} dx' \stackrel{}{=} \left(\lambda_Y + c \lambda_Z\right) \frac{\lambda_{X}}{2} ,
		\end{align}
		where $(a)$ follows from the properties of the limit (limit of the sum, is the sum of the limits and limit of the product is product of the limits), $(b)$ follows from Lemma $1$.  
	\end{proof}
	By making the appropriate substitutions it is straightforward to obtain the result of Theorem 1.
	\vspace{-1mm} 
	\section{Proof of Theorem 2.}
	\label{C}
	Here, we utilize the necessity of Karush-Kuhn-Tucker (KKT) conditions to obtain the optimal solution.	We start by writing the Lagrangian function for problem \eqref{eq:17}, which is given by
	\begin{align}
	\label{eq,c1:18}
	\mathcal {L}& (P_m, \mu_m,\lambda) \nonumber \\
	&= \frac{1}{M}\sum_{m=1}^M P_m E_{m-1} + \sum_{m=1}^M \mu_m P_m+ \lambda(E_M-\epsilon),
	\end{align}
	where $\mu_m$ for $m=1,\ldots,M$  and $\lambda$ are the Lagrangian multipliers. The KKT conditions are:
	\begin{enumerate}[label=\textbf{C\theenumi},itemsep=2pt,parsep=2pt,topsep=2pt,partopsep=2pt]
		\item $\frac{\partial{\mathcal{L}}}{\partial {P_m}}=0,~ m=1,\ldots,M$,
		\label{c1:c1}	
		\item $\mu_m \geq 0,~ m=1,\ldots,M$,
		\label{c1:c2}
		\item $\mu_mP_m=0,~ m=1,\ldots,M$,
		\label{c1:c3}
		\item $E_M-\epsilon=0$.
		\label{c1:c4}
	\end{enumerate}
	
	We can write the derivative of the Lagrangian function $\mathcal {L} (P_m, \mu_m,\lambda)$ with respect to the power $P_m$ as
	\begin{align}
	\label{eq,c1:19}
	&\frac{\partial \mathcal {L} (P_m, \mu_m,\lambda)}{\partial P_m }= \prod_{i=1}^{m-1} \frac{\phi_i \psi(\eta_i)}{P_i^2} \\ \nonumber &- \sum_{i=1}^{M-m} \frac{2 P_{i+1} \prod_{j=1}^{m+i-1} \phi_j \psi(\eta_j)}{P_m^3 \prod_{j=1, j \neq m}^{m+i-1} P_j} - \frac{2\lambda \prod_{i=1}^{M} \phi_i \psi(\eta_i)}{P_m ^3 \prod_{i=1, i \neq m}^{M} P_i^2},
	\end{align}
	where $\psi(\eta_i) = \frac{\eta_i^2 +1}{2 \eta_i^2}$. To solve the problem we relax the non-negativity requirement for the power terms, i.e. $\mu_m = 0$ $\forall m$. We write \ref{c1:c3} for $m = M$ as 
	\begin{align}
	\label{eq,c1:20}
	\frac{\partial \mathcal {L} (P_m, \mu_m,\lambda)}{\partial P_M } = \prod_{i=1}^{M-1} \frac{\phi_i \psi(\eta_i)}{P_i^2} - \frac{2\lambda \prod_{i=1}^{M} \phi_i \psi(\eta_i)}{P_M ^3 \prod_{i=1}^{M-1} P_i^2}.
	\end{align}
	Equating \eqref{eq,c1:20} to zero, we obtain the transmit power at the $M^{th}$ ARQ round as
	\vspace{-1mm} 
	\begin{align}
	\label{eq,c1:21}
	P_M=\sqrt[3]{ 2 \lambda \phi_M \psi(\eta_M)}.
	\end{align}
	Similarly substituting $m=M-1$ in \eqref{eq,c1:19}, we can rewrite \ref{c1:c1} for $m= M-1$ as
	\begin{align}
	&\frac{\partial \mathcal {L} (P_m, \mu_m,\lambda)}{\partial P_{M-1} }=\prod_{i=1}^{M-2} \frac{\phi_i \psi(\eta_i)}{P_i^2} - \frac{2 P_{M} \prod_{j=1}^{M-1} \phi_j \psi(\eta_j)}{P_{M-1}^3 \prod_{j=1}^{M-2} P_j} \nonumber  \\ & \qquad  \qquad \qquad \qquad - \frac{2\lambda \prod_{i=1}^{M} \phi_i \psi(\eta_i)}{P_{M-1} ^3 \prod_{i=1, i \neq M-1}^{M} P_i^2} . \label{eq,c1:23}
	\end{align}
	Equating \eqref{eq,c1:23} to zero leads to
	\begin{align}
	\label{eq,c1:41}
	P_{M\!-\!1\!}\!=\!\sqrt[3]{\!2 P_M \phi_{\!M\!-\!1\!} \psi(\eta_{M\!-\!1})\! \!+ \! \!\frac{2\lambda \! \phi_M  \phi_{M-1} \psi(\eta_M) \psi(\eta_{M-1})}{P_M^2}}.
	\end{align}
	
	We can continue this procedure for all $m \in \{1, \ldots, M\}$ and the results can be summarized as:
	\vspace{-1mm} 
	\begin{align}
	\nonumber
	P_M&=f(\lambda),\\
	\label{c1:b}
	P_{M-1}&=f(\lambda, P_M) ,\\
	\label{c1:a}
	P_{M-2}&=f(\lambda, P_M, P_{M-1}), \nonumber \\
	\vdotswithin{P_{1}}	\vspace{-1mm}  \nonumber\\ 	
	P_{1}&=f(\lambda, P_M, \ldots, P_3, P_2) .
	\vspace{-1mm} 
	\end{align}
	\vspace{-1mm} 
	By utilizing a method that is
	similar to the backward substitution approach \cite[App. C.2]{Boyd-Vandenberghe-04}, we can obtain a relationship between the power terms $P_m$ as follows: first, by substituting $2 \lambda \phi_M \psi(\eta_{M}) = P_M^3$ (see (\ref{eq,c1:21})) in \eqref{eq,c1:41} (or equivalently in \eqref{c1:b}) we evaluate $P_{M-1}$ as $P_{M-1} = \sqrt[3]{2 \phi_{M-1} \psi(\eta_{M-1}) P_M}$. Next, $P_{M-2}$ is evaluated by substituting  $\sqrt[3]{2 \lambda \phi_M \psi(\eta_{M})} = P_M$ and  $\sqrt[3]{2 \phi_{M-1} \psi(\eta_{M-1}) P_M} = P_{M-1}$ in \eqref{c1:a}. By continuing this procedure we can express the optimal transmit power in the $m^{th}$ round as
	\begin{align}
	\label{eq,c1:44}
	P_m=\sqrt{3 \phi_m \psi(\eta_{m}) P_{m+1}}, \; m\in\{1, \ldots, M-1\}.
	\end{align}
	Based on \eqref{eq,c1:44}, we can easily verify now that the obtained power values $P_m$ are all positive. Further, since $P_M$ is a function of $\lambda$ (see \eqref{eq,c1:21}) and using \eqref{eq,c1:44}, it is clear that each $P_m$ is a function of $\lambda$. Thus, all that remains is to compute the Lagrangian multiplier $\lambda$. For this purpose, we utilize the outage constraint in (\ref{eq:17}) \ref{c1:c4}. First, we substitute $P_m$ for $m =1,\ldots,M$ in \eqref{eq:16} to obtain $E_M$ as
	\vspace{-3mm}
	\begin{align}
	\label{eq,c1:29}
	E_m= \prod_{m=1}^M \frac{\phi_m \psi(\eta_{m})}{P_m^2}=\epsilon, 	\vspace{-1mm} 
	\end{align}
	\vspace{-1.5mm}
	where $P_m$ is given by
	\vspace{-2.5mm} 
	\begin{align}
	\label{eq,c1:32}
	\rho_m = 3^{o(m)} (2 \lambda)^{p(m)} \prod_{i=m}^{M} \left(\phi_i \psi(\eta_{i})\right)^{q(m)}.
	\end{align}
	In \eqref{eq,c1:32}, we compute the exponents: $o(m)=\sum_{i=1}^{M-m} \frac{1}{3^i}$, $p(m)=\frac{1}{3^{M-m+1}} \textrm{ and } q(m)= \frac{1}{3^{i-m+1}}$. Finally, we compute $\lambda$ by equating $E_M$ to the outage target $\epsilon$ based on \ref{c1:c4}, i.e.,
	\vspace{-1mm}
	\begin{align*}
	\label{eq,c1:39}
	\lambda \! \! = \! \! \!\left(\! \!\frac{\frac{1}{\epsilon} \prod_{m=1}^{M} \phi_m \psi(\eta_{m})}{\prod_{m=1}^{M} \! \! 2^{2 p(m)} 3^{2 o(m)} \!  \prod_{m=1}^{M} \left( \! \! \prod_{i=m}^{M} \! \left(\phi_i \psi(\eta_{i})\right)^{q(m)}\right)^2}\! \! \right)^{ \! \!k(m)} \! \! \! \! \! \! \!,
	\end{align*}
	where $k(m) = \sum_{m=1}^{M} \frac{-2}{3^{M-m+1}}$.
\end{appendices}
\bibliographystyle{IEEEtran}
\bibliography{IEEEabrv,Papers}

\begin{thebibliography}{10}
\providecommand{\url}[1]{#1}
\csname url@samestyle\endcsname
\providecommand{\newblock}{\relax}
\providecommand{\bibinfo}[2]{#2}
\providecommand{\BIBentrySTDinterwordspacing}{\spaceskip=0pt\relax}
\providecommand{\BIBentryALTinterwordstretchfactor}{4}
\providecommand{\BIBentryALTinterwordspacing}{\spaceskip=\fontdimen2\font plus
\BIBentryALTinterwordstretchfactor\fontdimen3\font minus
  \fontdimen4\font\relax}
\providecommand{\BIBforeignlanguage}[2]{{%
\expandafter\ifx\csname l@#1\endcsname\relax
\typeout{** WARNING: IEEEtran.bst: No hyphenation pattern has been}%
\typeout{** loaded for the language `#1'. Using the pattern for}%
\typeout{** the default language instead.}%
\else
\language=\csname l@#1\endcsname
\fi
#2}}
\providecommand{\BIBdecl}{\relax}
\BIBdecl

\bibitem{Popovski2017}
P.~Popovski, J.~J. Nielsen, C.~Stefanovic, E.~d.~Carvalho, E.~Strom, K.~F.
  Trillingsgaard, A.~Bana, D.~M. Kim, R.~Kotaba, J.~Park, and R.~B. Sorensen,
  ``Wireless access for ultra-reliable low-latency communication: Principles
  and building blocks,'' \emph{IEEE Network}, vol.~32, no.~2, pp. 16--23, March
  2018.

\bibitem{Schulz}
P.~Schulz, M.~Matthe, H.~Klessig, M.~Simsek, G.~Fettweis, J.~Ansari, S.~A.
  Ashraf, B.~Almeroth, J.~Voigt, I.~Riedel, A.~Puschmann, A.~Mitschele-Thiel,
  M.~Muller, T.~Elste, and M.~Windisch, ``Latency critical {IoT} applications
  in 5{G}: Perspective on the design of radio interface and network
  architecture,'' \emph{IEEE Commun. Mag.}, vol.~55, no.~2, pp. 70--78, Feb
  2017.

\bibitem{Goldsmith}
A.~Goldsmith, \emph{{Wireless Communications}}.\hskip 1em plus 0.5em minus
  0.4em\relax Cambridge, UK: Cambridge University Press, 2005.

\bibitem{Lee}
S.~Lee, W.~Su, S.~Batalama, and J.~D. Matyjas, ``Cooperative decode-and-forward
  {ARQ} relaying: Performance analysis and power optimization,'' \emph{IEEE
  Trans. on Wireless Commun.}, vol.~9, no.~8, pp. 2632--2642, Aug 2010.

\bibitem{Makki}
B.~Makki, T.~Eriksson, and T.~Svensson, ``{On the Performance of the Relay ARQ
  Networks},'' \emph{{IEEE} Trans. Veh. Technol.}, vol.~65, no.~4, pp.
  2078--2096, April 2016.

\bibitem{LTS}
J.~N. Laneman, D.~N.~C. Tse, and G.~W. Wornell, ``{Cooperative Diversity in
  Wireless Networks: Efficient Protocols and Outage Behavior},'' \emph{{IEEE}
  Trans. Inf. Theory}, vol.~50, no.~12, pp. 3062--3080, Dec 2004.

\bibitem{Tumula2}
T.~V.~K. Chaitanya and E.~G. Larsson, ``Optimal power allocation for hybrid
  {ARQ} with chase combining in i.i.d. {R}ayleigh fading channels,''
  \emph{{IEEE} Trans. Commun.}, vol.~61, no.~5, pp. 1835--1846, May 2013.

\bibitem{5754756}
------, ``{Outage-Optimal Power Allocation for Hybrid ARQ with Incremental
  Redundancy},'' \emph{IEEE Trans. on Wireless Commun.}, vol.~10, no.~7, pp.
  2069--2074, July 2011.

\bibitem{Larsson_2}
{T. V. K. Chaitanya, T. Le-Ngoc and E. G. Larsson}, ``{Energy-Efficient Power
  Allocation for HARQ Systems},'' \emph{Handbook of Research on Next Generation
  Mobile Communication Systems}, 2015.

\bibitem{Larsson}
P.~Larsson, L.~K. Rasmussen, and M.~Skoglund, ``{Analysis of Rate Optimized
  Throughput for Large-Scale MIMO-(H)ARQ Schemes},'' in \emph{Proc., IEEE
  Globecom}, Dec 2014, pp. 3760--3765.

\bibitem{Dosti2}
E.~Dosti, M.~Shehab, H.~Alves, and M.~Latva-aho, ``{Ultra reliable
  communication via CC-HARQ in finite block-length},'' in \emph{European Conf.
  on Networks and Commun. (EuCNC)}, June 2017, pp. 1--5.

\bibitem{Dosti}
E.~Dosti, U.~L. Wijewardhana, H.~Alves, and M.~Latva-aho, ``{Ultra reliable
  communication via optimum power allocation for type-I ARQ in finite
  block-length},'' in \emph{Proc., IEEE Intl. Conf. on Communications (ICC)},
  May 2017, pp. 1--6.

\bibitem{Laneman}
J.~N. Laneman, ``Limiting analysis of outage probabilities for diversity
  schemes in fading channels,'' in \emph{Proc., IEEE Globecom}, vol.~3, Dec
  2003, pp. 1242--1246.

\bibitem{Boyd-Vandenberghe-04}
S.~Boyd and L.~Vandenberghe, \emph{Convex Optimization}.\hskip 1em plus 0.5em
  minus 0.4em\relax Cambridge, UK: Cambridge University Press, 2004.

\end{thebibliography}
%
%
%
%
\end{document}